\documentclass{article}[12pt]

\usepackage{amssymb}
\usepackage{graphicx}
\usepackage[applemac]{inputenc}
\usepackage[english]{babel}
\usepackage[a4paper]{geometry}
\usepackage{color}
\usepackage{amsthm}
\usepackage{amsmath,mathtools}
\usepackage{amssymb}
\usepackage{enumerate}
\usepackage{xspace}
\usepackage{subcaption}
\usepackage[linesnumbered,ruled]{algorithm2e}
\usepackage{url}
\usepackage{authblk}

\newcommand{\problem}{FMD}
\newcommand{\problemb}{FmD}

\newcommand{\buono}{good}
\newcommand{\cattivo}{bad}

\newcommand{\concat}{\cdot}

\newtheorem{theorem}{Theorem}
\newtheorem{definition}{Definition}
\newtheorem{lemma}{Lemma}

\newtheorem{proposition}{Proposition}

\title{String factorisations with maximum or minimum dimension}

 \author[1]{A.~Monti\thanks{monti@di.uniroma1.it}}
 \affil[1]{\small Department of Computer Science, Sapienza University of Rome, Italy.}

 \author[2]{B.~Sinaimeri\thanks{sinaimeri@inria.fr}}
 \affil[2]{\small INRIA Grenoble, UMR CNRS 5558 - LBBE, Universit\'e Lyon 1, France.}

\begin{document}
\maketitle

\begin{abstract}
In this paper we consider two  problems concerning string factorisation. Specifically given a string $w$ and an integer $k$  find a factorisation of $w$ where each factor has length bounded by $k$ and has the minimum (the \problemb\ problem) or the maximum (the \problem\ problem)  number of different factors.  The \problemb\ has been  proved to be NP-hard even if $k=2$  in \cite{Schmid2016} and for this case we provide a $3/2$-approximation algorithm.  The \problem\ problem, up to our knowledge has not been considered in the literature. We show that this problem is NP-hard  for any $k\geq 3$. In view of this we propose  a $2$-approximation algorithm (for any $k$) an exact exponential algorithm.  We conclude with some open problems.
\end{abstract}

\textbf{Keywords:}\\
String factorisation, NP-hard problems, Approximation algorithms

\section{Introduction}

Combinatorial properties of strings of symbols from a finite alphabet have been studied in many different fields such as computer science, mathematics, physics, biology (see for example \cite{Bafna96,Karp93,Bulteau14,Fernau2015}).  In particular a problem that has gained much attention is the \emph{equality-free} string \emph{factorisation} with bounded width \cite{Condon2008,Condon2015} which arises in bioinformatics and is motivated by the problem of gene synthesis \cite{STEMMER1995,Jayaraj2005}. A \emph{factorisation} $F$ of  \emph{size} $l$ of a string $w$ is any tuple $(u_1, u_2 \ldots, u_l)$ of substrings of $w$ (called \emph{factors}),  that satisfies $w = u_1 \concat u_2 \concat \ldots u_l$. The \emph{width} of a factorisation is the maximum length of its factors.   An equality-free factorisation is a factorisation where \emph{all} the factors are different. Several variations of this problem have been considered in literature leading to interesting combinatorial problems with applications in various areas as bioinformatics, pattern matching or data compression (see for example \cite{Fernau2015,Schmid2016}). If we keep the equality-free property then two variants may be considered: either we require the factorisation to have large size or to have small width \cite{Schmid2016}. Both of these problems are NP-hard even if the alphabet is binary as shown in \cite{Condon2015,Schmid2016}. If the equality-free property is dropped, and hence repetition of factors is allowed, we can consider the converse problem and ask for a factorisation that is highly repetitive, meaning that has the smallest number of different factors \cite{Schmid2016}.  %In \cite{Schmid2016} two variants have been considered: find a factorisation with the minimum number of different factors and having either a large size or small width. 
It seems then natural  to consider the versions of this problem resulting by the interplay of two variables: the width $k$ and the number of different factors $d$ (that we call the \emph{dimension}).  Specifically we consider the following two problems.
 
\bigskip
\noindent
\textit{PROBLEM:  Factorisation of Minimum Dimension (\problemb )}\\
\noindent
\textit{INPUT:} A string $w$ in $\Sigma^n$ and an integer $k$.\\
\noindent
\textit{OUTPUT:} A factorisation of $w$ of width at most $k$ and of minimum dimension.

\bigskip
\noindent
\textit{PROBLEM:  Factorisation of Maximum Dimension (\problem )}\\
\noindent
\textit{INPUT:} A string $w$ in $\Sigma^n$ and an integer $k$.\\
\noindent
\textit{OUTPUT:}  A factorisation of $w$ of width at most $k$ and  of maximum dimension.

\bigskip

The \problemb\ is NP-hard. Indeed, its decisional version is studied in \cite{Schmid2016} as the \emph{Minimum Repetitive Factorisation Width} problem. In that paper it is proved to be  NP-complete even for factorizations of width $k=2$ (but the size of the alphabet must be unbounded) and  an FPT algorithm with respect to the size of the dimension is provided.  

The \problem\ problem is the dual version of the \problemb\ problem and  up to our knowledge,  has not been considered in the literature.

In Section~\ref{sec:preliminaries} we provide some definitions that we use in the rest of the paper. In  Section~\ref{sec:FmD} we consider the \problemb\ problem and for $k=2$ we provide a $3/2$-approximation algorithm.  Concerning the \problem\ problem  we show in Section~\ref{sec:Np-hardness} that if the size of the alphabet $\Sigma$ is not constant the problem is NP-hard even for $k=3$. Thus, we propose in Section~\ref{sec:exact_Algorithms},  an exact algorithm of time  $O(nk |\Sigma|^k 4^{|\Sigma|^k})$. Notice that if $k$ and the size of $\Sigma$ are both constant then the problem can be solved in linear time. Furthermore, in Section~\ref{sec:Approximation} we propose a $2$-approximation algorithm (for any $k$) for the \problem\ problem. Finally in Section~\ref{sec:Conclusions} we conclude with some open problems.

\section{Preliminaries}\label{sec:preliminaries}

In this paper we consider strings of length $n$ over some finite alphabet $\Sigma$. For $w \in \Sigma^n$, we denote by $|w|$ its  length and  by $w[i:j]$ the subtring of $w$ starting at position $i$ and ending at position $j$, if $i=j$ we indicate by $w[i]$ the symbol at position $i$.  The concatenation of two strings $u, v$ is denoted by $u \concat v$. When $u$ is a substring of $w$ we write $u \subseteq w$.  A \emph{factorisation} $F$ of a string $w$ is any tuple $(u_1, u_2 \ldots, u_l)$ of substrings  of $w$ with $l\geq 1$,  that satisfies $w = u_1 \concat u_2 \concat \ldots u_l$.  The substrings $u_i$ are called \emph{factors} of $w$.   A substring of length \emph{at most } $k$ is called a $k$-factor.   An $1$-factor is called \emph{singleton}.  The \emph{width} of a factorization is the maximum length of its factors.  A factorisation in $k$-factors is called \emph{ $k$-factorisation}.  Furthermore, $D(F)$ stands for the set of factors in $F$ and its cardinality is called the \emph{dimension} of $F$, denoted by $d(F)$. Notice that, due to multiplicity, the number of factors in $F$ can be larger than $d(F)$.

%%%%%%%%%%%%%%%%%%%%%%%%%%%%%%%%%%%%%%%%%%%%%%%%%%%%%%%%%%%%
%%%%%%%%%%%%%%%%%%%%%%%%%%%%%%%%%%%%%%%%%%%%%%%%%%%%%%%%%%%%

\section{An approximation algorithm for the \problemb\ problem}\label{sec:FmD}
The \problemb\ problem is proved to be NP-hard in \cite{Schmid2016} even for $k=2$ (but the size of the alfabet must be unbounded). In the following we present a $3/2$- approximation algorithm for the case $k=2$.  Notice that given a string $w$ on $m$ symbols and an integer $k$ a trivial algorithm that outputs a factorisation $F$ where each factor is a singleton, is a $k$-approximation for the \problemb\ problem. Indeed, $d(F)=m$ and any  other $k$-factorisation will have at least $m/k$ factors. 

\subsection{A $3/2$-approximation algorithm for the case $k=2$}\label{subsec:approx}

Here we show that by  modifying the trivial $2$-approximation algorithm we obtain  a $3/2$-approximation algorithm. 

\begin{definition}\label{def:good_factor}
Given a string $w$ a factor $f=ab$  is said \emph{\buono} if there exist a 2-factorisation $F$ of $w$ such that there is no factor in $F$ different from $ab$ that contains $a$ or $b$. A factor of size $2$ that is not \buono\  is said \emph{\cattivo}.
\end{definition}

For example in the string $w=acbacba$ the factor $cb$ is good but the factor $ac$ is not. We are now ready to define our algorithm.  The algorithm starts with a factorisation $F$ of $w$ containing only singletons. Let  $A(w)$ be the set of all \buono\ factors in $w$. While there is  a \buono\ factor  $ab$ in $A(w)$ it modifies $F$ by merging all the adjacent $1$-factors $a$, $b$ into the $2$-factor $ab$.  Then it updates  $A(w)$ by removing all the factors that contain either  $a$ or $b$.  When this operation is not possible anymore it ends outputting $F$. The pseudocode  is given in Algorithm~\ref{algo:3/2-approx}.

\begin{algorithm}
    \SetKwInOut{Input}{Input }
    \SetKwInOut{Output}{Output}
    
    \underline{function  GreedyMinFactorisation} $(w)$\;
    \Input{A string $w \in \Sigma^n$}
    \Output{A list $F$ of $2$-factors $[u_1,\ldots, u_t]$ with $u_1 \concat \ldots \concat u_t=w$ }
    $F \leftarrow [w[1],\ldots, w[n]]$ \;
    $A \leftarrow$ the good factors of $w$\;
    \While{$A\neq \emptyset$}{
    {Choose a factor $ab$ in $A$\;}
    {Update $F$ by substituting all the subsequent factors $w[i]=a, w[i+1]=b$ with the factor $w[i:i+1]=ab$\;
    $A\leftarrow A - \{xy \in A: \{x,y\}\cap \{a,b\}\neq \emptyset \}$\;}
     }
     return $F$\;
    \caption{A greedy algorithm for finding a  2-factorisation with minimum dimension}
    \label{algo:3/2-approx}
\end{algorithm}

Clearly the algorithm produces a valid $2$-factorisation of $w$. However, this may not be an optimal solution. Indeed, let $w$ be the string $a_1b_1a_1c_1 a_2b_2a_2c_2 \ldots a_{t} b_{t} a_{t}c_{t}$  with $|\Sigma|=3t$ and $n= 4t$. Notice that $A(w) =\emptyset$ and thus, the algorithm  produces a factorization consisting of all the $3t$ singletons. An optimal factorization that will take the $2$-factors $a_ib_i$ and $a_ic_i$ for all $1\leq i \leq t$ will have dimension $2t$. This shows that the approximation factor of Algorithm~\ref{algo:3/2-approx} is at least $3/2$. We prove in the next theorem that this is the real approximation ratio of our algorithm.

\begin{theorem}
The algorithm \emph{GreedyMinFactorisation} is a $3/2$-approximation algorithm for the \problem\ problem with $k=2$.
\end{theorem}

\begin{proof}
Consider a string $w \in \Sigma^n$.  Let $S$ be the set of symbols that appear in $w$ and $m=|S|$. We denote by $S_{good}$ the set of symbols appearing in some maximal good factor in $A(w)$ and $S_{bad} = S- S_{good}$ the ones that do not appear. We denote by $m_{good}= |S_{good}|$ and $m_{bad}=|S_{bad}|$. Clearly, $m_{good}+ m_{bad}=m$.

Let $F^G$ be the factorisation produced from the greedy algorithm and let $F^*$ be an optimal factorisation.  We prove the following two results:

\begin{enumerate}
\item $d(F^G) \leq  m - \frac{m_{good}}{4} $ 
\item $d(F^*)\geq   \frac{2}{3} \bigl(m - \frac{m_{good}}{4}\bigr)$
\end{enumerate}

Notice that from $(1)$ and $(2)$ we have that $\frac{d(F^G)}{d(F^*)} \leq \frac{3}{2}$. 

\paragraph{Proof of $(1)$} The greedy takes two types of factors singletons and good factors. Let $\lambda^G_{good}$ be  factors of size two that are good and $\lambda^G_{singl}$ the number of factors of size 1. To prove (1) we show that:

\begin{enumerate}
\item[(i)] $\lambda^G_{singl}=m - 2 \lambda^G_{good}$
\item[(ii)] $\lambda^G_{good}\geq \frac{m_{good}}{4}$
\end{enumerate}

Notice that (1) follows immediately by (i)-(ii) by observing that $d(F^G)=\lambda^G_{singl}+\lambda^G_{good}$.  The equality in (i)  follows  by the definition of a good factor and line $6$ of Algorithm~\ref{algo:3/2-approx}. To prove (ii) we show that the algorithm makes at least $m_{good}/4$ steps. To this purpose, it is sufficient to  show that at each step of the algorithm the number of symbols that appear in $A=A(w)$ (\textit{i.e.} good symbols) decreases by at most $4$.  Consider a generic step of Algorithm~\ref{algo:3/2-approx} and let $f=ab$ the chosen good factor. By definition of good factor  there is at most one other factor in  $A$ containing $a$ ($b$) and this factor is of the form $xa$ ($by$).  Thus, at line 7 of the algorithm there are at most $3$ factors that are removed from $A$,  containing at most $4$ different symbols.

\paragraph{Proof of $(2)$} 

We first define a canonical form for a factorisation $F$. 

\begin{definition}\label{def:canonical}
We say that $F$ is in a \emph{canonical} form if the followings hold:
\begin{itemize}
\item[(1)]  for any factor $ab$ in $F$ none of the factors in $\{a,b,ba,bb,aa\}$ belongs to $F$.
\item[(2)] for any good factor $f=ab$ in $F$, there is no factor $f' \neq f$ in $F$ such that $f$ and $f'$ share a symbol.
\item[(3)] for any bad factor $f=ab$ in $F$, there is at least another bad factor $f' \neq f$ in $F$ such that $f$ and $f'$ share a symbol.
\end{itemize}
\end{definition}

\begin{lemma}\label{lem:optimal_canonical}
Any string $w$ has a canonical optimal factorisation $F$.
\end{lemma}
\begin{proof}
Let $F$ be an optimal factorization of $w$.  We consider three cases:

\paragraph{Case 1: } Suppose $F$ does not satisfy item (1) of Definition~\ref{def:canonical} then there exist in $F$ a factor $f=ab$ and a non empty set of factors $R \subseteq \{a,b,ba,bb,aa\}$. Then we  define a new factorisation $F'$ starting from $F$ and decomposing  all the factors $f'$ in $F$  with $f'\in R\cup\{ab\}$ into singletons. As $R$ is not empty  the number of different factors does not increase. Thus, $F'$ is still optimal.

\paragraph{Case 2:}  Suppose $F$ does not satisfy item (2) of Definition~\ref{def:canonical} then there exists a good factor $f=ab$ which shares symbols with other factors in $F$ and let $R$ be the set of these factors.  We define a new factorisation $F'$ starting from $F$, removing either $a$ or $b$ from each factor in $R$ different from $ab$. Finally we merge all the singletons $a$ and $b$ into $ab$. This operation is possible  as $ab$ was a good factor.We have that $F'$ is still a factorization and its dimension does not increase. Thus, $F'$ is still optimal.

\paragraph{Case 3:}  Suppose $F$ does not satisfy item (3) of  Definition~\ref{def:canonical} and let $f=ab$ be a bad factor such that there is no other bad factor in $F$ containing $a$ or $b$.  From the two previous items we can assume that the singletons $a$, $b$ are not in $F$ and no other good factor contains $a$ and $b$.   Thus, no other factor in $F$ contains $a$ or $b$.  By  Definition~\ref{def:good_factor} $f$ must be a good factor contradicting our initial hypothesis.

The proof concludes by observing that the three cases can be satisfied simultaneously. 
\end{proof}

From Lemma~\ref{lem:optimal_canonical} we can assume $F^*$ is canonical. The optimum will have three type of factors and let: $\lambda^*_{good}$ be the number factors of size two that are good, $\lambda^*_{bad}$ be the number of factors of size two that are bad and $\lambda^*_{singl}$ be the number  factors of size 1.  Observe that only good symbols may appear in good factors, whereas bad factor or singletons may contain both good and bad symbols. Let $m^{1}_{good}$ be the number of good symbols appearing in good factors in $F$ and $m^{2}_{good} +m^{2}_{bad}$ be the number of good and bad symbols appearing in bad factors. We show that:

\begin{enumerate}
\item[(i)] $\lambda^*_{singl}= m - m^{1}_{good} - m^{2}_{good} - m^{2}_{bad}$
\item[(ii)] $\lambda^*_{good} = \frac{m^{1}_{good}}{2}$
\item[(iii)] $\lambda^*_{bad} \geq \frac{2}{3} \bigl(m^{2}_{good} +m^{2}_{bad}\bigr)$
\end{enumerate}

The inequalities  (i) and (ii) are straightforward.  The inequality  (iii) is proved by the following lemma.

\begin{lemma}\label{lem:bad_factors}
Let $F$ be any canonical optimal factorisation of $w$, $f_1, \ldots, f_t$ all the bad factors in $F$ and  $S$ the set of symbols that appear in these factors.  It holds  that  $|S| \leq 3t/2$.
\end{lemma}

\begin{proof}
To prove the lemma we construct a graph $G=(V,E)$ where $V=\{f_1, \ldots, f_r\}$ and $\{f_i,f_j\} \in E$ if the factors $f_i$ and $f_j$ share a symbol. We partition the set of vertices according to the following procedure: We start with $G_1=G$. At step $i\geq 1$ we choose a vertex $v_i$ of maximum degree and let $R_i$  be the set of vertices containing $v_i$ together with its adjacents  in $G_i$.  Then  $G_{i+1}$ is obtained from $G_{i}$ by deleting the vertices in $R_i$ and the edges incident to them. The procedure ends when there are no more vertices left (\textit{i.e.} $G_i$ is empty). Let $R_1, \ldots R_k$ be the sequence of vertices produced by this procedure. Let  $|R_i|=r_i$ for any $1\leq i \leq k$. From item (iii) of Definition~\ref{def:canonical} the graph $G$ has no isolated vertex. Thus, let $j$ be the last index for which $r_j \geq 2$. Notice that $1\leq j \leq k$. For each $i \leq j$ the number of symbols that appear in the set $R_i$ is at most $r_i+1$. Indeed, $v_i$ contributes by at most $2$ symbols and each adjacent vertex shares a symbol with $v_i$ and hence the contribution of each adjacents to the set of symbols is at most $1$.  For any $i \geq j+1$ the vertices $v_i$ are isolated in $G_i$ and thus could theoretically contribute by $2$ symbols. However, as the graph $G$ was connected, $v_i$ must necessarily be adjacent to some vertex in $\cup_{l=1}^{i-1} R_i$ and thus each isolated vertex contributes by at most one new symbol. Summarizing we have that

$$
|S| \leq \sum_{i=1}^{j} (r_i+1) + \sum_{i=j+1}^{k} r_i=   \sum_{i=1}^{k} r_i + \sum_{i=1}^{j} 1 = t + j \leq \frac{3t}{2}
$$
\noindent
Where the last inequality follows by observing that for every $i \geq j$, $r_i \geq 2$  and thus  $j\leq t/2$  .
 
\end{proof}
\noindent
In conclusion using (i)-(iii) we have that:

\begin{eqnarray}\nonumber
d(F^*) &=& \lambda^*_{singl}+\lambda^*_{good}+ \lambda^*_{bad} \\ \nonumber
&\geq & m - \frac{1}{3} \biggl( m^{1}_{good} + m^{2}_{good} + m^{2}_{bad}\biggr) - \frac{1}{6}m^{1}_{good}  \\ 
&\geq & m  - \frac{1}{3} m - \frac{1}{6}m^{1}_{good} \\
&\geq & \frac{2}{3} \biggl(m - \frac{m_{good}}{4}\biggr).
\end{eqnarray}

The inequality (1) follows by the definition of good and bad symbols and point (2) of Definition~\ref{def:canonical} while inequality (2) follows by simply observing that $m^{1}_{good} \leq m_{good}$.

\end{proof}

We conclude this section by observing that the \emph{GreedyMinFactorisation} algorithm runs in $O(kn)$ computational time. Indeed, the set of good factors $A(w)$ can be computed in $O(n)$ by using for example a hash table and the second part of the algorithm requires $O(kn)$ time  as each symbol of $w$ is considered at most $k$ times.

%%%%%%%%%%%%%%%%%%%%%%%%%%%%%%%%%%%%%%%%%%%%%%%%%%%%%%%%%%%%
%%%%%%%%%%%%%%%%%%%%%%%%%%%%%%%%%%%%%%%%%%%%%%%%%%%%%%%%%%%%

\section{NP-hardness of the \problem\ problem}\label{sec:Np-hardness}

In this section we prove the  hardness of the \problem\ problem.  To this purpose we show the NP-completeness of its decision version.\\

\bigskip
\noindent
\textit{PROBLEM:  $k$-Factorisation of Maximum  Dimension ($k$-\problem\ )}\\
\noindent
\textit{INPUT:} A string $w$ in $\Sigma^n$ and an integer $d$.\\
\noindent
\textit{QUESTION:} Does there exists  a factorisation $F$ of $w$ of width at most $k$ such that $d(F)\geq d$ ? \\
 
\noindent 
The reduction is from the $3$-Dimensional -Matching problem that is known to be one of Karp's 21 NP-complete problems \cite{Karp1972}.  Let $X, Y$ and $Z$ be finite, disjoint sets, with $|X|=|Y|=|Z|$ and let $T$ be a subset of $X \times Y \times Z$. That is, $T$ consists of triples $(x, y, z)$ such that $x \in X, y \in Y$, and $z \in Z$.  $M \subseteq T$ is a $3$-dimensional matching if the following holds: for any two distinct triples $(x_1, y_1, z_1) \in M$ and $(x_2, y_2, z_2) \in M$, we have $x_1 \neq x_2, y_1 \neq  y_2$, and $z_1 \neq z_2$.  

\bigskip
\noindent
\textit{PROBLEM:  $3$-Dimensional-Matching problem  ($3$-DMP)}\\
\noindent
\textit{INPUT:} Sets   $X, Y$ and $Z$ each of size $l$ and a set  $T\subseteq X \times Y \times Z$.\\
\noindent
\textit{QUESTION:} Does there exists  a $3$-dimensional matching $M\subseteq T$ with $|M|=l$ ?
 
We are ready to prove the main result of this section.

\begin{theorem}
The $k$-\problem\ problem is NP-complete for any  $k\geq 3$. 
\end{theorem}

\begin{proof}
Let $T\subseteq X \times Y \times Z$ and  $l=|X|=|Y|=|Z|$ be  an instance of  $3$D-Matching. We create an instance of $k$-\problem\ as follows. We denote $|T|=t$.

\begin{itemize}
\item For each triple $(p_i,q_i,r_i) \in T$, with $1\leq i \leq t$, we create 
\begin{itemize}
\item the substring of length 6:  $u_i=t^1_i p_i q_i r_i t^2_i f_i$. 
\item the substring of length 4: $x_i=a^1_i t^1_i p_i a^2_i$.
\item the substring of length 4: $y_i=b^1_i r_i t^2_i  b^2_i$.

\end{itemize}

\item For each pair $(p,q)$ such that there exists $(p_i,q_i,r_i) \in T$, with $p_i=p, q_i=q$ we create  the substring of length 4: $z_{pq}=c^1_{pq} p q c^2_{pq}$.
\item For each pair $(q,r)$ such that there exists $(p_i,q_i,r_i) \in T$, with $q_i=q, r_i=r$ we create  the substring of length 4: $z_{qr}=c^1_{qr} q r c^2_{qr}$.

\end{itemize}

We consider the string $w= u_1\concat u_2 \concat \ldots \concat u_t \concat x_1\concat x_2 \ldots \concat x_t \concat y_1 \concat y_2 \concat \ldots \cdot y_t \concat z_{pq}  \concat \ldots \concat z_{qr}$. Notice that the length of the string $w$ is $6t + 4t +4t + 4t_1 + 4t_2$ where $t_1 \leq t, t_2 \leq t$ are the number of the pairs $(p,q)$ and $(q,r)$, respectively. %Thus, the length of the string $w$ is upper bounded by $22t$. 
The alphabet $\Sigma$ is given by the union of the sets $X, Y, Z$ and the set containing the symbols $t^1_i, t^2_i, f_i, a^1_i, a^2_i,b^1_i, b^2_i, c^1_{pq}, c^2_{pq}, c^1_{qr}, c^2_{qr}$. Hence the size of $\Sigma$ is given by $3l+7t+2t_1+2t_2 \leq 3l+11t$. Finally to define the  instance of $k$-\problem\ we set $d= 2l+10t+ 3t_1 +3t_2$.

We show now that $T$ has a $3$D-Matching of size $l$ if and only if $w$ has a $k$-factorization $F$, with $k \geq 3$ such that $d(F)\geq d$. \\

 \paragraph{$(\leftarrow)$} Let $M\subseteq T$ be a $3$D-matching with $|M|=l$.  We construct the following  factorization $F$ of $w$:

\begin{enumerate}
\item[(i)] For all $i$ for which $(p_i,q_i,r_i) \in M$ we split $u_i=t^1_i p_i q_i r_i t^2_i f_i$ in 6 factors: $t^1_i$,$p_i$, $q_i$, $r_i$, $ t^2_i$,$ f_i$. Notice that these factors appear exactly once as $M$ is a $3D$-matching. This contributes to the dimension of $F$ with $6l$ different factors. 

\item[(ii)] For all $i$ for which $(p_i,q_i,r_i) \not\in M$ we split $u_i=t^1_i p_i q_i r_i t^2_i f_i$ in  4 factors: $t^1_i$,$p_i q_i r_i$, $ t^2_i$,$ f_i$. Notice that these factors appear exactly once as all the triples in $T$ are different and $t^1_i,   t^2_i$ and $ f_i$ are different from the ones in item (i) as the set of triples considered in (i) and (ii) are disjoint.  This contributes  with $4 (t-l)$ different factors.

\item[(iii)] For all $1\leq i\leq n$ we split  $x_i=a^1_i t^1_i p_i a^2_i$  in 3 factors: $a^1_i$, $t^1_i p_i$, $a^2_i$. Notice that these factors are all different among them and also different from the ones in the previous items. This contributes  with $3t$ different factors.

\item[(iv)]  For all $1\leq i\leq n$ we split   $y_i=b^1_i r_i t^2_i  b^2_i$ in 3 factors: $b^1_i$, $r_i t^2_i$, $b^2_i$. Similarly to the previous item this contributes  with $3t$ different factors.

\item[(v)] For all $z_{pq}$ we split into 3 factors: $c^1_{pq}$, $p q$, $c^2_{pq}$. This contributes  with $3t_1$ different factors.

\item[(vi)] For all $z_{qr}$ we split  into 3 factors : $c^1_{qr}$, $q r$, $c^2_{qr}$. This contributes  with $3t_2$ different factors.
\end{enumerate}

Thus we have a factorisation $F$ of width $k=3$ and  $d(F)= 6l + 4 (t-l) + 3t+3t+ 3t_1 +3t_2 =2l+10t+ 3t_1 +3t_2=d$.

 \paragraph{$(\rightarrow)$} 
 
 We first define a canonical form for the  factorization $F$ of our string $w$.
 
 \begin{definition}\label{def:canonical1}
We say that $F$ for $w$ is in  \emph{canonical} form if the followings hold:
\begin{itemize}
\item[(1)]  the symbols $f_i, a^1_i, a^2_i, b^1_i, b^2_i, c^1_{pq}, c^2_{pq}, c^1_{qr}, c^2_{qr}$ appear as  $1$-factors in $F$.
\item[(2)] for all $i\in [t]$  the substrings $x_i$ and $y_i$ are factorized with the factors  $a^1_i$, $t^1_i p_i$, $a^2_i $ and $b^1_i$, $r_i t^2_i$,  $b^2_i$, respectively.
\item[(3)] for all pairs $(p,q), (q,r)$ that appear in the triples of $T$, the substrings $z_{pq}$ and $z_{qr}$ are factorized with the factors  $c^1_{pq}$, $p q$, $c^2_{pq}$  and  $c^1_{qr}$, $q r$, $c^2_{qr}$, respectively. 
\end{itemize}
\end{definition}

\begin{lemma}\label{lem:optimal_canonical1}
The  string $w$ has a canonical  factorisation $F$ of maximum dimension.
\end{lemma}
\begin{proof}
Let $F$ be a factorization of $w$ of maximum dimension. If $F$ is canonical we are done otherwise we consider three cases:

\paragraph{Case 1} Suppose $F$ does not satisfy item (1) of Definition~\ref{def:canonical1}. Then there is in $F$, a factor  $w_j$ of length at least 2 which contains at least one symbol $s$ in $\{f_i, a^1_i, a^2_i, b^1_i, b^2_i, c^1_{pq}, c^2_{pq}, c^1_{qr}, c^2_{qr}\}$.  We can define a new factorisation $F'$ from $F$ by decomposing $w_j=w^1_j \concat s\concat w^2_j$. Notice that $w_j$ contributes to the dimension of the factorisation by at most one. Moreover, the factorization $F'=w_1 \concat \ldots \concat w^1_j \concat s\concat w^2_j \concat \ldots \concat w_m$ contains the  factor $s$ which never appeared in $F$ (as $s$  appears only once in $w$). Hence, the dimension of the new factorization does not decrease.

\paragraph{Case 2} Suppose $F$ satisfies item (1) but does not satisfy item  (2) of Definition~\ref{def:canonical1}. Suppose first there exists an $i$ such that the string $x_i=a^1_i t^1_i p_i a^2_i$ is not split in $a^1_i$, $t^1_i p_i$, $a^2_i$.  However, $F$ satisfies item (1), $a^1_i$ and $a^2_i$ must necessarily be singletons and thus the only possibility left is that $x_i$ is split in $a^1_i$, $t^1_i$,$p_i$, $a^2_i$. We consider the new factorisation $F'$ where  we substitute  the two factors $t^1_i$,$p_i$, with the single factor $t^1_i p_i$.  Notice that $w$ contains exactly two occurrences of $t^1_i$: one in $x_i$ and one in $u_i=t^1_i p_i q_i r_i t^2_i f_i$. As all $u_j$ are separated by $f_j$, using item (1) of Definition~\ref{def:canonical1}, we have that in $F'$ (and also in $F$) $u_i$ must be split in such a way that the first factor is of the form $t^1_i\concat v$ (with $v \subseteq p_i q_i r_i t^2_i$). We consider the factorisation $F''$ for which we substitute the factor   $t^1_i\concat v$ with two factors $t^1_i$ and $v$. Notice that if $v=\emptyset$ then $F''=F'$. It remains to show that $d(F'') \geq d(F)$.  Notice that the only difference between $F$ and $F''$ is that $F$ contains $t^1_i\concat v$ and $p_i$, while $F''$ contains $t^1_i p_i$ and $v$.
We need to consider the following cases:

\begin{itemize}

\item  If $v=\emptyset$ or $v=p_i$ then $d(F'')= d(F)$.

\item Otherwise $p_i q_i\subseteq v\subseteq p_i q_i r_i t^2_i$ then we can split even more $v$ to ensure  that the dimension of the factorization does not decrease. Thus we can define $F'''$ which is obtained from $F''$ by substituting the factor $v$ with $p_i$ and $v'$ (where $v=p_i \concat v'$). Clearly $d(F''')= d(F)$.
\end{itemize}

It can be shown similarly that there exists a  factorisation of maximum dimension in which $y_i$ is decomposed as $b^1_i$, $r_i t^2_i$,  $b^2_i$ for all $i$. And this concludes the proof.

\paragraph{Case 3} Suppose $F$  satisfies items (1) and (2) of Definition~\ref{def:canonical1} but does not satisfy (3).  We follow the same argument as in the previous case.  There exists   a pair $(p,q)$ such that $z_{pq}=c^1_{pq} p q c^2_{pq}$ is not split in $c^1_{pq}$, $pq$ and $c^2_{pq}$. 
However, by item $(1)$  the $c^1_{pq}$ and $c^2_{pq}$ must necessarily be singletons and thus the only possibility left is that $z_{pq}$ is split in $c^1_{pq}$, $p$, $q$ and $c^2_{pq}$.  We define a new factorisation $F'$ starting from $F$ and substituting the two singletons $p$,$q$, with the single factor $pq$. We need to consider the following cases:
 \begin{itemize}

 \item $pq$ is not a factor of $F$ and at least one between $p$ and $q$ appear more than once in $F$ (and thus they still appear in $F'$). Then clearly $d(F')$ cannot be smaller than $d(F)$.
  
 \item $pq$ is a factor of $F$, then we define $F''$ starting from $F'$ and splitting the factor $pq$ of $F$ in $p$ and $q$. Clearly,   $d(F'')= d(F)$.

 \item Otherwise $pq$ is not a factor of $F$ and $p$ and $q$ appear both exactly once in $F$ (and thus they do not appear in $F'$).  Thus $d(F')=d(F)-1$. However, there exists at least one triple $(p_i,q_i,r_i) \in T$ for which $p_i=p, q_i=q$. Then there are two possibilities: 
\begin{itemize}
 \item $u_i$ is split such that $p, q$ appear together in a single factor $v_ipqv'_i$ with at least one among $v_i, v'_i$ non empty. Then we define  $F''$ from $F'$ by splitting the factor $v_ipqv'_i$ into $v_i$, $p$, $q$ and $v'_i$. Notice that even if $v_ipqv'_i$ appeared exactly once and $v_i$ and $v'_i$ are not new, we have that $d(F'')= d(F')-1+2=d(F)$;

 \item   $u_i$ is split such that $p, q$ appear in two distinct factors that can only be $t^1_i p$ and $qrv'_i$. Then we can assume from item  (2) of Definition~\ref{def:canonical1}  that in $F$ (and in $F'$ also) the factor $t^1_i$ never appears. Thus, we define $F''$ from $F'$ by splitting $t^1_i p$ into $t^1_i$ and  $p$. Notice that $d(F'')= d(F')+1= d(F)$. 
 \end{itemize} 
 
 \end{itemize}
  It can be shown similarly that there exists a  factorisation of $w$ of maximum dimension  in which $z_{qr}$ is decomposed as $c^1_{qr}$, $q r$, $c^2_{qr}$, for all pairs $(q,r)$.
\end{proof}

Let $F$ be a factorization of $w$ of maximum dimension with $d(F) \geq d$. From Lemma~\ref{lem:optimal_canonical1} we can assume $F$ is canonical. Thus, we have that the substrings  $x_i$ and $y_i$ of $w$ contribute to the dimension of $F$ by at most $3t+3t=6t$ factors (by item (2) of Definition~\ref{def:canonical1}), and the subtrings $z_{pq}$ and $z_{qr}$ contribute by $3t_1+3t_2$  (by item (3) of Definition~\ref{def:canonical1}).  It remains to quantify the contribution of the substrings $u_i$, for all $1\leq i\leq t$.  We partition the set of $u_i$s into two subsets $A$ and $B$ where $A$ is the set of  $u_i$s that are decomposed in $F$ in  exactly $6$ factors $t^1_i$,$p_i$, $q_i$, $r_i$, $ t^2_i$,$ f_i$. Let $t'=|A|$.  Notice that any  $u_i=t^1_i p_i q_i r_i t^2_i f_i$ in $B$ can be decomposed in $4$ or $5$ factors (from item (1) of Definition~\ref{def:canonical1}). In both cases the contribution, of each of these substrings to the dimension of $F$ is at most $4$ as for all $i$ $p_iq_i$ and $q_i r_i$ belong to $F$  (by item (2) of Definition~\ref{def:canonical1}). Hence in total the substrings in $B$ contribute by at most $4(t-t')$. The contribution of the substrings $u_i \in A$ is at most $3t'+3 \min \{t',l\}$ as the number of distinct factors $p_i$, $q_i$ and $r_i$ is bounded by both $|A|$ and $|X|+|Y|+|Z|=3t$ (the total number of symbols). Hence, $d(F) \leq 6t + 3t_1+3t_2 + 4(t-t')+ 3t'+3 \min \{t',l\} $. If $t' < l$ or $t' > l$ then $d(F) < 2l + 10t + 3t_1 + 3t_2$.   However, we know that $d(F) \geq 2l+10t+3t_1+3t_2$. Thus, the only possibility is that $t'=l$ and each $u_i$ must contribute by exactly $6$ to the dimension. Thus, the triples $(p_i,q_i,r_i)$ corresponding to $u_i \in A$ form a $3D$-matching. This concludes the proof.

\end{proof}

%%%%%%%%%%%%%%%%%%%%%%%%%%%%%%%%%%%%%%%%%%%%%%%%%%%%%%%%%%%%
%%%%%%%%%%%%%%%%%%%%%%%%%%%%%%%%%%%%%%%%%%%%%%%%%%%%%%%%%%%%

\section{A $2$-approximation algorithm for the \problem\ problem}\label{sec:Approximation}
Here we provide a greedy  $2$-approximation algorithm for the \problem\ problem for any width $k$.  Given $w \in \Sigma^n$ and an integer $k$ the main idea is to read  $w$ in a sequential way and   each time we look for a new $k$-factor (\textit{i.e.} a factor that has not appeared before) that ends first. The pseudocode is given in Algorithm~\ref{algo:2-approx}. The list $F$ contains a factorisation of a prefix of $w$ and the set $D=D(F)$ is the set of the factors in $F$. The variable $p$ contains the smallest index for which the suffix $w[p:n]$ is yet not factorized.  Initially $p=1$.  At each step we search for a new $k$-factor in the $w[p:n]$ (\textit{i.e.} a factor not in $D$). If there is none then we partition the string $w[p:n]$ arbitrary into $k$-factors. %for example if $n-p+1=ak+b$ we take the first $b$ factors of length exactly $k$ and the last factor of length $b<k$. 
Otherwise, we consider among the new $k$-factors the one that finishes in the smallest index and in case of ties we choose the  one that starts first (\textit{i.e.} the longest one). Let $w[t:j]$ be the factor chosen, then if $p<t$ then we partition the string $w[p:t]$ arbitrary in $k$-factors.  

Clearly the Algorithm~\ref{algo:2-approx} produces a $k$-factorisation of $w$ which may not be optimum. Indeed, let  $w=aababa$ and let $k=2$.  The algorithm will produce the factorisation $F= a\concat ab \concat a \concat b \concat a  $ that has dimension 3 as the set of different factors is $\{a, b,  ab\}$. However, an optimal factorisatation is of dimension $4$ as $F^*=  aa\concat b \concat ab \concat a $.  The following theorem shows that this algorithm is a a $2$-approximation for the \problem\ problem and  runs in  $O(kn)$ computational time.

\begin{algorithm}
    \SetKwInOut{Input}{Input }
    \SetKwInOut{Output}{Output}

    \underline{function  GreedyMaxFactorisation} $(w,k)$\;
    \Input{A string $w \in \Sigma^n$ and an integer $k$}
    \Output{A list $F$ of $k$-factors $[u_1,\ldots, u_m]$ with $u_1 \concat \ldots \concat u_m=w$}
    $p \leftarrow 1$\;
    $F \leftarrow [\,]$ \;
    $D \leftarrow \emptyset$\; %\tcp{$D$ contains the set of different $k$-factors}
    \While{true}{
    $I \leftarrow$ all $k$-factors in $w[p:n]$ that do not belong in $D$\;
    \If{$I=\emptyset$}
      {
       exit\;
      }
     Let $w[t:j] \in I$ be the longest $k$-factor with smallest $j$\; 
     $D\leftarrow D\cup \{w[t:j]\}$\;
     \If{$p<t$}
     {
     Let $w[p,t-1]=s_1 \concat \ldots \concat s_r$ such that $\forall 1\leq l < r $ $|s_l| =k$ and $|s_r| \leq k$\;
     Append $s_1, \ldots,  s_r$ to $F$\;
     }
     Append $w[t:j]$ to $F$\;
     $p\leftarrow j+1$\;
     }
     \If{$p\leq n$}
     {
      Let $w[p,n]=s_1 \concat \ldots \concat s_r$ such that $\forall 1\leq l < r $ $|s_l| =k$ and $|s_r| \leq k$\;
      Append $s_1, \ldots,  s_r$ to $F$
     }
     return $F$\;
    \caption{A greedy algorithm for finding a  k-factorisation with maximum dimension}
    \label{algo:2-approx}
\end{algorithm}

\begin{theorem}\label{theo:2approx}
The algorithm \emph{GreedyMaxFactorisation} is a $2$-approximation for the \problem\  problem and runs in time $O(kn)$.
\end{theorem}

\begin{proof}

Let $w$ and $k$ be an instance of \problem\ problem. Clearly, \emph{GreedyMaxFactorisation} produces a $k$-Factorisation for the string $w$. We prove that it is a $2$-approximation.  Let $F$ and $D$ be the $k$-factorization and the set of $k$-factors produced at the end of \emph{GreedyMaxFactorisation} and let $d=d(F) = |D|$.
Notice that each factor in $D$ is implicitly associated to its first occurrence in $F$.  Consider any $k$-factorisation $F^*$ and let $D^*=D(F^*)$ and $d^*=d(F^*) = |D^*|$. Notice that it is sufficient to show that $d^* \leq 2 d$.  To this purpose we partition the set $D^*$ into two sets: the set $S_1$ of $k$-factors that were chosen  by the greedy algorithm and $S_2$ of the $k$-factors that were not. More formally, let $D^*= S_1 \cup S_2$ where $S_1 \subseteq D$ and $S_2 \cap D =\emptyset$. Each element of the set $S_2$ can be associated to its first occurrence $w[i:j]$ in $F^*$. %Thus, the set $S_2$ is a set of disjoint subtrings with fixed positions in $w$. 
Clearly $|S_1| \leq |D| =d$. It remains to show that $|S_2| \leq d$.

Given two subtrings $s=w[i:j]$ and $s'=w[i':j']$ of $w$, we say that $s$ \emph{ends within} $s'$ if $j \in [ i',j']$. 

\begin{proposition}\label{prop:withinproperty}
For any $k$-factor $v$ of the string $w$ such that $v \not \in D$ there exists a $k$-factor $u \in D$ such that $u$ ends within $v$. 
\end{proposition}
\begin{proof}
Let  $v=w[i:j]$ be a $k$-factor not belonging to $D$.  We search for a $k$-factor $u=w[i':j']$ in $D$ of maximum $j'$ such that  $j' \leq j$. Notice that such a string must exists  as at least $w[1:1]$ belongs to $D$. We prove that $j'\in [i,j]$. Suppose on the contrary that $j' < i$. Then after the $k$-factor $u=w[i':j']$ is added to $D$ at line 11, the algorithm continues with $p=j'+1$ at line 17. Notice that now $v=w[i:j] \in I$ and as $v \not \in D$ the algorithm must have chosen another string with $w[i'':j'']$ with  $j'+1\leq j'' \leq j$ which contradicts the maximality of $j'$.
\end{proof}

By Proposition~\ref{prop:withinproperty}  each $k$-factor $v=w[i:j]$ of $S_2$ can be associated to a $k$-factor  $u=w[i':j']$ in $D$. Moreover, the $k$-factor $u=w[i':j']$ of $D$ can be associated to at most one $k$-factor of $S_2$ since at most one $k$-factor of $S_2$ includes position $j'$. Thus, $|S_2| \leq |D| =d$. 

To show that \emph{GreedyMaxFactorisation} runs in $O(kn)$ time notice that inside the while loop we do not need to compute each time the set $I$. Indeed, we can implement the search for the new $k$-factor $w[t:j]$, that will be added in $D$, as follows:  at each position $j$ with  $p \leq j\leq n$ we search for the longest new $k$-factor that ends in $j$.  This means that starting from $j$ we will go back and consider at most $k$ positions. This gives a computational time of $O(kn)$.

\end{proof}

 We conclude this section by noticing that if $k$ is constant the \emph{GreedyMaxFactorisation} algorithm  runs in time linear in $n$  which is clearly optimum.

%%%%%%%%%%%%%%%%%%%%%%%%%%%%%%%%%%%%%%%%%%%%%%%%%%%%%%%%%%%%
%%%%%%%%%%%%%%%%%%%%%%%%%%%%%%%%%%%%%%%%%%%%%%%%%%%%%%%%%%%%

\section{An exact exponential algorithm for the \problem\ problem}\label{sec:exact_Algorithms}

In Section~\ref{sec:Np-hardness} we proved that the \problem\ problem is NP-hard even when $k=3$ and here we consider an exact algorithm for this problem. Notice that the exhaustive algorithm that considers all the possible $k$-factorisations of a string $w$ of length $n$  has a computational time bounded from below by  $k^{\Omega(n/k)}$ (resulting by the recurrence $T(n)=\sum_{j=1}^k T(n-j)$, known as $n$-Step  Fibonacci \cite{Wolfram}). In agreement with our NP-hardness result, when $k$ is constant the algorithm remains exponential.   In the following we give  an exact algorithm based on dynamic programming that runs in $O(nk |\Sigma|^k 4^{|\Sigma|^k})$ computational time. From this result we have that a constant bound on both  $k$ \emph{and} $|\Sigma|$ allows to solve the problem in polynomial time,  actually in linear time (\textit{i.e.} time $O(n)$). Notice  that  this is again in agreement with our NP-hardness result where the size of the alphabet is unbounded.

\begin{theorem}\label{theo:exact_algo}
The \problem\  problem can be solved in time $O(nk |\Sigma|^k 4^{|\Sigma|^k})$.
\end{theorem}

\begin{proof}
Given $w \in \Sigma^n$ and $k$  we start by introducing the algorithm that finds a set of $k$-factors $X^*$, for which there exists a $k$-factorisation of maximum dimension $F$ with $D(F)=X^*$.

Let $S$ be the set of all possible $k$-factors of $w$.  We define a dynamic programming matrix $T$ as follows. For every $1\leq i \leq n$ and every $X \subseteq S$ we define 

{\footnotesize
\[
T[i,X]=
\begin{dcases}
1, & \textrm{ if  there is a $k$-factorisation of  $w[1:i]$ for which the set of  $k$-factors is $X$} \\
0, & \textrm{otherwise}
\end{dcases}
\]
}
 
Notice that $T[n,X]=1$ if and only if  there exists a $k$-factorisation $F$ of $w$ with $D(F)=X$. Thus, $w$ has a $k$-factorisation of dimension $d$, for some integer $d$,  if and only if  $\bigvee_{X \in {\cal F}, |X|= d}  T[n,X]$. Furthermore we can find the set $X^*$ by looking for the $X$ of maximum cardinality for which$T[n,X]=1$. Hence, we can find $X^*$ by computing the values of all the $n2^{|S|}$ cells in $T$.  

Clearly, for every $X$, $T[1,X]=1$ if $X=\{w[1]\}$ and $0$ otherwise. For every $i>1$ notice that if $T[i,X]=1$ then there exists a  $j<i$ for which $w[j:i]$  belongs to $X$. Notice that as $X$ contains only strings of length at most $k$, we have that $j\geq \max\{1, i-k+1\}$. If $j=1$ then $T[i,X]=1$ if and only if $X=\{w[1:i]\}$. Otherwise let $j>1$, again as $T[i,X]=1$ we must have that for $w[1:j-1]$ there exists a $k$-factorisation $F$  such that either $D(F)=X$ or  $D(F)=X-\{w[j:i]\}$.  
Summarizing, for $i>1$ we have:

{\scriptsize
\begin{equation}\label{eq:recursion}
T[i,X]=\bigvee_{\max(1,i-k+1)<j\leq i}\Bigl( \bigl( w[j:i] \in X \bigr) \wedge  \bigl((j=1 \wedge |X|=1) \vee T[j-1,X] \vee T[j-1,X-\{w[j: i]\}\bigr) \Bigr)
\end{equation}
}
Note that we can compute the value of each cell $T[i,X]$ in time $O(k|X|)=O(k|S|)$. Thus we can obtain $X^*$ in time $O\left(nk|S|2^{|S|}\right)$.

So far we have simply computed the set $X^*$ of $k$-factors for which exists an optimal $k$-factorization $F$ with $D(F)=X^*$ (\textit{i.e.} a $k$-factorisation of maximum dimension). Finally, to find the factorisation $F$, we can use a simple procedure which, starting from the cell $T[n,X^*]$ and applying the recursion in (\ref{eq:recursion}) in the reverse order,  permits to find the sequence (in reverse order) of $k$-factors in $F$. This procedure requires again $O\left(nk|S|2^{|S|}\right)$.
Hence we can solve the \problem\ problem in $O\left(nk|S|2^{|S|}\right)$.

Notice now that as $S$ is a subset of all possible strings of length at most $k$ on $\Sigma$, thus  $|S|\leq \sum_{i=1}^{k}|\Sigma|^i\leq 2|\Sigma|^k$. Hence, our algorithm requires $O\left(nk|\Sigma|^k4^{|\Sigma|^k}\right)$ computational time.

  \end{proof}

\section{Conclusions and open questions}\label{sec:Conclusions}

In this paper we studied the \problemb\ and the \problem\ problem.  The first one has been proved to be NP-hard even if $k=2$ in \cite{Schmid2016} and  for this case  we provide a $3/2$-approximation algorithm. It remains an open problem  to find  a``good''  approximation algorithm that works for any $k$. Notice that the trivial algorithm that produces a factorisation containing only singletons is already a $k$-approximation.

Concerning the \problem\ problem  we showed that if the size of $\Sigma$ is not constant the problem is NP-hard even for $k=3$. In view of this we proposed  an exact algorithm of time  $O(nk |\Sigma|^k 4^{|\Sigma|^k})$. From this result follows that if $k$ and $|\Sigma|$ are both constant then the problem can be solved in linear time. Furthermore,  we proposed a $2$-approximation algorithm that works for any $k$. Many open questions remain.  Perhaps the most interesting one is whether  the  problem remains NP-hard when  $|\Sigma|$ is a constant (but the width is unbounded). Notice that despite the theoretical interest, this question has also a practical importance as the strings of DNA or RNA are on a fixed alphabet.  Concerning the $2$-approximation algorithm we tested it on random strings of length at most 16. In these experiments  the algorithm produces solutions that were always better than a $2$-approximation. Thus, an open question to see if the \emph{GreedyMaxFactorisation} algorithm  is tight or  the analysis of the approximation ratio can be improved.

%%%%%%%%%%%%%%%%%%%%%%%%%%%%%%%%%%%%%%%%%%%%%%%%%%%%%%%%%%%%%%%%%%%%%%%%%%%%%%%%%%%%%%%%%%%%%%%%%%%
%%%%%%%%%%%%%%%%%%%%%%%%%%%%%%%%%%%%%%%%%%%%%%%%%%%%%%%%%%%%%%%%%%%%%%%%%%%%%%%%%%%%%%%%%%%%%%%%%%%
\bibliographystyle{plain}
\bibliography{references_string}

\end{document}